\documentclass[conference]{IEEEtran}
\usepackage{amsfonts, amsmath, amsthm, tikz}
\usepackage[hidelinks]{hyperref}

\newtheorem{prop}{Proposition}
\newtheorem{thm}{Theorem}

\theoremstyle{definition}
\newtheorem{defn}{Definition}

\theoremstyle{remark}
\newtheorem{rem}{Remark}
\newtheorem{exa}{Example}

\newcommand{\F}{\mathbb F}
\newcommand{\mcC}{\mathcal C}

\newcommand{\smsum}{\textstyle\sum\limits}

\newcommand{\cir}[1]{\draw (#1) circle (.45);}
\newcommand{\fil}[1]%
{\fill[gray!50!white] (#1) circle (.45); \cir{#1};}
\newcommand{\rec}[2]%
{\draw[semithick] (#1-.45, #2-.45) rectangle (#1+.45, #2+.45);}
\newcommand{\fns}{\footnotesize}

\DeclareMathOperator{\Mat}{Mat}

\begin{document}

\title{Pseudoredundancy for the Bit-Flipping Algorithm}

\author{\IEEEauthorblockN{Jens Zumbrägel}
  \IEEEauthorblockA{%
    \textit{Faculty of Computer Science and Mathematics} \\
    \textit{University of Passau} \\
    Innstraße 33, 94032 Passau, Germany \\
    jens.zumbraegel@uni-passau.de}}

\maketitle

\begin{abstract}
  The analysis of the decoding failure rate of the bit-flipping
  algorithm has received increasing attention.  For a binary linear
  code we consider the minimum number of rows in a parity-check matrix
  such that the bit-flipping algorithm is able to correct errors up to
  the minimum distance without any decoding failures.  We initiate a
  study of this bit-flipping redundancy, which is akin to the stopping
  set, trapping set or pseudocodeword redundancy of binary linear
  codes, and focus in particular on codes based on finite geometries.
\end{abstract}

\section{Introduction}

The bit-flipping algorithm of Gallager~\cite{Gal63} is a simple but
effective iterative decoding method.  It has received increased
attention recently due to its usage in the post-quantum cryptography
scheme BIKE~\cite{BIKE}.  The idea of this cryptosystem is based on
decoding a moderate-density parity-check code for the legitimate
party, while hiding the sparse parity-check structure for the
attacker~\cite{M+13}.  A significant issue for this scheme is however
the analysis of the decoding failure rate.  Indeed, the asymptotic
analysis cannot be applied to concrete codes due to the presence of
cycles in the Tanner graph.

In the context of finite-length analysis of iterative decoding schemes
failure patterns have been studied such as stopping sets~\cite{D+02},
trapping sets~\cite{Ric03}, pseudocodewords~\cite{VK05} and absorbing
sets~\cite{Dol10}.  In this regard a \emph{pseudoredundancy} has been
considered as the minimum number of parity-checks that avoid these
failure patterns up to a certain extent.  In particular, the stopping
set redundancy~\cite{SV06}, the trapping set redundancy~\cite{L+09}
and the pseudocodeword redundancy for several channels~\cite{KS07,
  ZSF12} have been investigated.

In this work we propose a notion of bit-flipping pseudoredundancy of
a binary linear code as the minimum number of rows in a parity-check
matrix such that the bit-flipping algorithm is able to correct errors
up to the minimum distance without any decoding failures.  While the
concept seems to be very similar to the stopping set or pseudocodeword
redundancy, one significant difference is that the analysis is more
dependent on the structure of the columns rather than the rows of
a parity-check matrix.  Relevant for this work is also the notion
of expander codes~\cite{SS96}, as this allows to make rigorous
statements about the bit-flipping algorithm.

We initiate the study of the bit-flipping redundancy and present some
first results.  After stating the preliminaries in the next section,
including the definition of bit-flipping pseudoredundancy, we focus
in particular on codes based on finite geometries and an analysis
of $t$-error correction for small~$t$.

\section{Preliminaries}

\subsection{Parity-Check Codes}

We consider a binary linear code $\mcC \subseteq \F_2^n$ of length~$n$
given by a \emph{parity-check matrix} $H = (h_{ji}) \in
\Mat_{r \times n}(\F_2)$, i.e.\
\[ \mcC = \big\{ x \in \F_2^n \mid \smsum_{i=1}^n h_{ji} x_i = 0
  \text{ for all } j \big\} . \]
The matrix~$H$ is not required to be full-rank, so $r \ge n - k$
where $k = \dim \mcC$ is the code dimension.

The \emph{Tanner graph} associated to~$H$ is the bipartite graph with
variable nodes $v_1, \dots, v_n$ and check nodes $w_1, \dots, w_r$,
with~$v_i$ adjacent to~$w_j$ precisely when $h_{ji} = 1$.  We assume
this graph to be left-regular, so that each variable node has same
degree~$c$.  This means that each column $h_i$ of the matrix~$H$ has
constant weight~$c$.

\subsection{Bit-Flipping Algorithm}

The \emph{bit-flipping algorithm} is a simple iterative decoding method.
Let~$\mcC$ be a binary code of length~$n$ with parity-check matrix~$H$.
Given a received vector $y = x + e \in \F_2^n$ (where $x \in \mcC$ is
a codeword and $e \in \F_2^n$ is the error vector) with syndrome
$s = H y^T \in \F_2^r$ do:
\begin{enumerate}
\item for each~$i$ compute the number~$u_i$ of unsatisfied parity-check
  neighbors of variable node~$v_i$; if all $u_i \le \frac c 2$ stop,
  return~$e$
\item for those~$i$ with largest~$u_i$ perform a bit flip $e_i' = e_i
  + 1$ and update the syndrome $s' = s + h_i$, repeat from 1)
\end{enumerate}
Step~2) can be performed either step-by-step or in parallel.  Note
that if $u_i > \frac c 2$ there are more unsatisfied check neighbors
than satisfied ones and therefore the syndrome weight decreases.

There is an alternative, more combinatorial description that will
assist us for the subsequent arguments.  Let us identify vectors in
$\F_2^r$ with subsets of $[r] = \{ 1, \dots, r \}$ in an obvious way.
Consider thus the~$n$ columns of~$H$ as blocks $B_i \subseteq [r]$
each of size~$c$ in a set of size~$r$.  Given now a received vector
$y \in \F_2^n$ with syndrome $S \subseteq [r]$ do:
\begin{enumerate}
\item for each~$i$ compute $u_i = \vert B_i \cap S \vert$; if all
  $u_i \le \frac c 2$ stop
\item for largest~$u_i$ update syndrome $S' = S \,\triangle\, B_i$
  (symmetric difference), repeat from 1)
\end{enumerate}

Now we can define the pseudoredundancy for the bit-flipping algorithm.
It is well-known that a code of minimum distance~$d$ can correct up
to~$\frac {d-1} 2$ errors by maximum-likelihood decoding.

\begin{defn} Let $\mcC \subseteq \F_2^n$ be a binary linear code
  of minimum distance~$d$.  We define the (bit-flipping)
  \emph{pseudoredundancy}~$\rho$ of the code~$\mcC$ as the minimum
  number~$r$ of rows in a parity-check matrix for~$\mcC$ such that
  the corresponding bit-flipping algorithm corrects up to
  $\frac {d-1} 2$ errors. \end{defn}

In this work we deal with left-regular Tanner graphs, thus we require
the parity-check matrix to have constant column weight, although it is
conceivable to relax this condition.  In case there is no parity-check
matrix with this property the pseudoredundancy is understood to
be~$\infty$.  Presently, we do not know a concrete example of a code
with infinite pseudoredundancy.

\subsection{Expander Codes}

The Tanner graph specifies a $(c, d, \alpha)$-\emph{expander} code if
the graph is left-regular of degree~$c$ and any subset of $t \le d$
variable nodes has more than~$\alpha t$ neighboring check nodes.

Sipser and Spielmann~\cite[Thm.~7, Thm.~10]{SS96} have proven the
following results (see also~\cite[Ch.~12]{HLW06}).

\begin{thm}\label{mindist} A $(c, d, \frac c 2)$-expander code has
  minimum distance greater than~$d$. \end{thm}

This result is not hard to show.  Indeed, take a vector of weight
$t \le d$ corresponding to a set of~$t$ variable nodes.  These are
connected by $c t$ edges to more than $\frac c 2 t$ neighboring check
nodes.  So one of those check nodes is connected to only one
of these variable nodes and thus is unsatisfied.

\begin{thm} For a $(c, d, \frac 3 4 c)$-expander code the step-by-step
  bit-flipping algorithm corrects up to $\frac d 2$ errors. \end{thm}

The proof idea is the following.  Consider an error vector~$e$ of
weight~$t$ given by a set of~$t$ variable nodes connected to~$u$
unsatisfied check nodes.  Suppose that $t \le d$ and there are~$s$
satisfied neighbors.  Then $u + s > \frac 3 4 c t$ and $c t \ge
u + 2 s$, whence $u = 2 (u + s) - (u + 2 s) > \frac c 2 t$.  Since 
these~$t$ variable nodes have $c t$ edges, there is hence one with
more than $\frac c 2$ unsatisfied neighbors, in which case we could
flip the corresponding variable.  However, the algorithm may decide to
flip a different variable that is not erroneous.  Still the method
only fails to decode if the error weight increases to $t = d$ at a
later stage, in which case $u > \frac c 2 d$.  But if we have
$t \le \frac d 2$ errors at the start, this cannot occur since
$u \le \frac c 2 d$ is decreasing.

\section{Codes Based on Finite Geometries}

A class of expander codes can be obtained from finite geometries.
Consider a Tanner graph with variable nodes $v_1, \dots, v_n$ and
check nodes $w_1, \dots, w_r$.  As before, for a variable node~$v_i$
let $B_i \subseteq [r]$ denote the indices~$j$ of neighboring
check nodes~$w_j$, which is a set of size~$c$.

We may view the~$B_i$ as blocks or \emph{lines} in a point set of
size~$r$.  These are said to form a \emph{partial geometry} if
distinct lines intersect in at most one point, or equivalently, two
points lie on at most one line.  For the parity-check matrix this
means that there is no rectangle of $1$'s, and for the Tanner graph to
contain no four-cycle.

\begin{exa} Let $q = p^m$ be a prime power and consider a projective
  plane of order~$q$, which has $n = q^2 + q + 1$ points.  There
  are~$n$ lines, each having $c = q + 1$ points, and any two lines
  intersect in one point. \end{exa}

Note that~$t$ distinct lines in a partial geometry have at most
${t \choose 2}$ intersection points, so their union has at least
$c t - {t \choose 2}$ elements.  It follows that these define
$(c, t, \alpha)$-expander codes where $\alpha = c - \frac{t-1} 2$.

So $\alpha > \frac c 2 t$ if and only if $t - 1 < c$, and
$\alpha > \frac 3 4 c t$ if and only if $t - 1 < \frac c 2$.
Hence we can take maximum $d = c$ for the minimum distance result and
maximum $d = \lceil \frac c 2 \rceil$ for the error-correction
(leading to correcting up to $\frac {c+1} 4$ errors).  We can improve
the error-correction capability as follows.

\begin{prop}\label{partial} For a code based on a partial geometry of
  constant block size~$c$ the bit-flipping algorithm corrects up to
  $\frac c 2$ errors. \end{prop}

\begin{proof} We argue that the (step-by-step or parallel) bit-flipping
algorithm never flips a non-erroneous variable node as long as $t \le
\frac c 2$.  Suppose that, say, the first~$t$ variable nodes are
erroneous, and consider the syndrome set \[ S = B_1 \,\triangle\, \dots
\,\triangle\, B_t \] composed out of the~$t$ blocks~$B_i$.  Then every
block~$B_i$ intersects~$S$ in at least $c - t + 1$ points, while
some other block intersects~$S$ in at most~$t$ variables.  Since
$2 t \le c$ we have $c - t + 1 > t$ and the algorithm chooses to flip
one of the error bits. \end{proof}

We note that the bit-flipping algorithm has been studied for
finite-geometry codes by Kou, Lin and Fossorier~\cite[Sec.~IV-A]%
{KLF01}, however their focus was less on a rigorous analysis.

\begin{exa} For $q = 2^m$ the parity-check code corresponding to the
  projective plane on $n = 4^m + 2^m + 1$ points with block size
  $c = 2^m + 1$ has dimension $4^m - 3^m + 2^m$ and minimum distance
  $2^m + 2$, cf.\ \cite[Sec.~13.8]{MS77}.  So Prop.~\ref{partial}
  shows that the pseudoredundancy of this code is $\rho \le n$.

  Similarly, the punctured Euclidean plane constitutes a finite
  geometry on $n = 4^m - 1$ points and~$n$ lines each having $c = 2^m$
  points.  This parity-check code has dimension $4^m - 3^m$ and
  minimum distance $2^m + 1$, cf.\ \cite[Sec.~III-A]{KLF01}, hence
  the pseudoredundancy again is $\rho \le n$ by Prop.~\ref{partial}.

  Concretely, for $q = 4$ we obtain binary linear $[15, 7, 5]$ and
  $[21, 11, 6]$ codes with finite pseudoredundancy. \end{exa}

Prop.~\ref{partial} can easily be generalized for blocks of size~$c$
such that distinct blocks intersect in at most~$s$ points, in which
case the bit-flipping algorithm corrects up to $t \le \frac c {2 s}$
errors (in fact, any~$t$ with $(2 t - 1) s < c$ suffices).

\begin{rem} A parity-check matrix as above is the incidence matrix of
  a “partial $(c, s)$ design” in the sense of \cite[Def.~7.1]{ZSF12}.
  It has been shown that the various minimum pseudoweights satisfy
  $w \ge 1 + \frac c s$ (cf.\ \cite[Thm.~7.3]{ZSF12}), so that
  $w \ge 1 + c$ in the case of a partial geometry. 

  Therefore, these pseudoweight bounds correspond to the
  error-correcting capability of the bit-flipping algorithm. \end{rem}

\subsection{The Hamming and Simplex Codes}

We now provide some results which correspond to those
in~\cite[Sec.~7]{ZSF12}.

A binary linear code has minimum distance at least~$3$ if and only if
the columns of a parity-check matrix are distinct.  In such a case
it is easy to see that the bit-flipping algorithm corrects one error.
We obtain the following.

\begin{prop} The $[n, n - m, 3]$ Hamming code where $n = 2^m - 1$
  has pseudoredundancy $\rho \le n$. \end{prop}

\begin{proof} Take a circulant parity-check matrix for the Hamming
  code, which has~$n$ rows and constant column weight. \end{proof}

\begin{prop} The $[n, m, 2^{m-1}]$ simplex code where $n = 2^m - 1$
  has pseudoredundancy $\rho \le \frac 1 6 n (n - 1)$. \end{prop}

\begin{proof} Consider the Hamming code dual of the simplex code and
  take as rows of the parity-check matrix all codewords of weight~$3$.
  This matrix has $r = \frac 1 6 n (n - 1)$ rows, since there are
  $\frac 1 2 n (n - 1)$ vectors of weight~$2$ and each row covers
  three of these.  The blocks $B_1, \dots, B_n$ corresponding to the
  columns have size $c = \frac {n-1} 2$ and form a partial geometry.
  Indeed, any two-element set $\{ i, j \}$ is covered by (at
  most) one Hamming codeword of weight~$3$, hence the intersection
  $B_i \cap B_j$ has at most one element.  So by applying
  Prop.~\ref{partial} we can correct up to $\frac {n-1} 4 =
  \frac {d-1} 2$ errors. \end{proof}

We note that taking a circulant parity-check matrix, as in the
previous proof, would not suffice in this case.

\subsection{Eigenvalue Analysis}

The notion of expander graph is closely related to the \emph{spectral
  gap} of the largest and second-largest eigenvalue of the incidence
matrix.  In this regard the following eigenvalue bound by
Tanner~\cite{Tan01} is of interest.

Consider a binary linear code of length~$n$ with parity-check matrix~$H$
such that the Tanner graph is connected, left-regular of degree~$c$ and
right-regular of degree~$d$.  Denote by~$\lambda_1$ and~$\lambda_2$ the
largest and second-largest eigenvalues, respectively, of the matrix
$H^T H$.

\begin{thm} With the above assumptions the code's minimum distance
  satisfies \[ w \,\ge\, n \!\cdot\! \frac {2 c - \lambda_2}
    {\lambda_1 - \lambda_2} \,. \] \end{thm}

Vontobel and Koetter~\cite{VK04} have generalized this bound to the
minimum pseudoweight for the additive white Gaussian noise channel.
The spectral gap has in turn influence on the expansion property, as
proven earlier by Tanner~\cite{Tan84}.

\begin{thm} Any subset of~$t$ variable nodes has at least~$u$
  neighboring check nodes, where \[ u \,\ge\, \frac {c^2 t}
    {(\lambda_1 - \lambda_2) t / n + \lambda_2} \,. \] \end{thm}

The proofs of these results depend crucially on the spectral theorem
for the symmetric matrix $H^T H$.

\begin{exa} Consider the projective plane of order~$q$ on $n =
  q^2 + q + 1$ points and the corresponding incidence matrix~$H$.
  We have $c = d = q + 1$ and the eigenvalues are $\lambda_1 =
  (q + 1)^2$ and $\lambda_2 = q$.  In this case the eigenvalue
  bound is \[ w \ge n \!\cdot\! \frac {q + 2} {q^2 + q + 1} =
    q + 2 \,, \] and the number of neighbors of~$t$ variable nodes is
  \[ u \ge \frac {c^2 t} {t + \lambda_2} = \frac {(q + 1)^2 t}
    {t + q} \,. \] We note that $u > \frac c 2 t$ if and only if
  $\frac {q + 1} {t + q} > \frac 1 2$, so we may apply
  Thm.~\ref{mindist} with $d = q + 1$. \end{exa}

\section{Decoding Failure}

In order to examine more precisely when the (step-by-step)
bit-flipping algorithm succeeds or fails, we study the correction
of~$t$ errors for small~$t$.

\subsection{Two Errors}

We consider $t = 2$ errors corresponding to blocks $B_1, B_2$ of
size~$c$ and syndrome $S = B_1 \,\triangle\, B_2$.  If the bit-flipping
algorithm chooses a right block to flip, say~$B_1$, the new syndrome
is $S' = B_2$ and the decoding succeeds.  On the other hand, if it
selects some other block~$B_3$ then the new syndrome $S' = B_1
\,\triangle\, B_2 \,\triangle\, B_3$ is of size at most~$c$, the
syndrome size in the former case.  Then the decoder necessarily
fails, since either $S'$ is some block different from $B_1, B_2$,
or the syndrome weight decreases further.

Denote by $s_{12}, s_{13}, s_{23}$ the intersection
sizes of $B_1 \cap B_2$, $B_1 \cap B_3$, $B_2 \cap B_3$ and by
$s_{123}$ the size of $B_1 \cap B_2 \cap B_3$.  Then the blocks~$B_1$
and~$B_2$ intersect~$S$ in $c - s_{12}$ points, while the other
block~$B_3$ intersects~$S$ in $s_{13} + s_{23} - 2 s_{123}$ points.
Therefore, the bit-flipping algorithm chooses a right block and thus
succeeds, precisely if $c > s_{12} + s_{13} + s_{23} - 2 s_{123}$.

\begin{exa} If the intersection of two blocks has at most~$s$ points,
  then the algorithm corrects two errors provided that $c > 3 s$.  So
  in the case of a partial geometry we require $c > 3$.  In the
  borderline case $c = 3$, bit-flipping corrects two errors if and
  only if the partial geometry does not have a \emph{triangle}, i.e.\
  three lines with each two intersecting in distinct points. \end{exa}

\subsection{Three Errors}

As the number~$t$ of errors increases the possible block
configurations become more involved.  In the following we discuss the
case of a partial geometry.  From Prop.~\ref{partial} we know that the
bit-flipping algorithm corrects $t = 3$ errors if $c > 5$.

In the case $c = 5$ consider blocks $B_1, B_2, B_3$ that form a
triangle.  Then the syndrome set $S = B_1 \,\triangle\, B_2
\,\triangle\, B_3$ has~$9$ points and intersects each block~$B_i$
in~$3$ points.  Suppose there is some other block~$B_4$ intersecting
each of $B_1, B_2, B_3$, then the algorithm may choose that block to
flip in which case it fails.  Indeed the new syndrome
$S' = B_1 \,\triangle\, B_2 \,\triangle\, B_3 \,\triangle\, B_4$ is of
size~$8$ yet not of a form $B_1 \,\triangle\, B_2$ as required for the
second-to-last step of a correct decoding. It is not hard to see that
bit-flipping for $c = 5$ corrects three errors precisely if there is
no configuration of four lines with each two intersecting in distinct
points.

\subsection{More Errors}

Consider a partial geometry in which not every pair of lines
intersects.  Then with growing~$t$ it becomes less likely that
given~$t$ blocks intersect pairwise.  Therefore a union of~$t$ blocks
often has more than $c t - {t \choose 2}$ elements, so the
corresponding codes have better expansion properties.

Also it occurs that even if a non-erroneous block has been selected,
the decoder may still decode successfully.

\begin{exa} Let the Tanner graph be left-regular of degree $c = 5$.
  Suppose that four errors occur corresponding to blocks $B_1, B_2,
  B_3, B_4$, intersecting pairwise except for $B_1, B_3$ and $B_2, B_4$.
  Then the syndrome set $S = B_1 \,\triangle\, B_2 \,\triangle\, B_3
  \,\triangle\, B_4$ is of size~$12$.  If there is another block~$B_5$
  intersecting~$S$ in three points, a run of the bit-flipping
  algorithm may select this block yet still decode successfully, see
  Fig.~\ref{bitflip}. \end{exa}

\begin{figure}
  \caption{Example run of the bit-flipping algorithm.  The white and
    gray disks are the satisfied and unsatisfied parity-checks,
    respectively.  The squares represent a block of parity-checks
    that is being affected by a bit-flip, while the given number of
    total unsatisfied parity-checks is decreasing.}\label{bitflip} \bigskip
  \begin{center} \begin{tikzpicture}[scale=.45]
    \cir{0, 4} \fil{1, 4} \fil{2, 4} \fil{3, 4} \cir{4, 4}
    \fil{0, 3} \fil{4, 3} \fil{0, 2} \fil{4, 2} \fil{0, 1} \fil{4, 1}
    \cir{0, 0} \fil{1, 0} \fil{2, 0} \fil{3, 0} \cir{4, 0}
    \cir{1.5, 2} \cir{2.5, 2}
    \rec{2}{4} \rec{0}{2} \rec{4}{2} \rec{1.5}{2} \rec{2.5}{2}
    \node at (2, -1.3) { \fns $12$ };
  \end{tikzpicture} \qquad
  \begin{tikzpicture}[scale=.45]
    \cir{0, 4} \fil{1, 4} \cir{2, 4} \fil{3, 4} \cir{4, 4}
    \fil{0, 3} \fil{4, 3} \cir{0, 2} \cir{4, 2} \fil{0, 1} \fil{4, 1}
    \cir{0, 0} \fil{1, 0} \fil{2, 0} \fil{3, 0} \cir{4, 0}
    \fil{1.5, 2} \fil{2.5, 2}
    \rec{0}{0} \rec{1}{0} \rec{2}{0} \rec{3}{0} \rec{4}{0}
    \node at (2, -1.3) { \fns $11$ };
  \end{tikzpicture} \qquad
  \begin{tikzpicture}[scale=.45]
    \cir{0, 4} \fil{1, 4} \cir{2, 4} \fil{3, 4} \cir{4, 4}
    \fil{0, 3} \fil{4, 3} \cir{0, 2} \cir{4, 2} \fil{0, 1} \fil{4, 1}
    \fil{0, 0} \cir{1, 0} \cir{2, 0} \cir{3, 0} \fil{4, 0}
    \fil{1.5, 2} \fil{2.5, 2}
    \rec{0}{0} \rec{0}{1} \rec{0}{2} \rec{0}{3} \rec{0}{4}
    \node at (2, -1.3) { \fns $10$ };
  \end{tikzpicture} \bigskip\\
  \begin{tikzpicture}[scale=.45]
    \fil{0, 4} \fil{1, 4} \cir{2, 4} \fil{3, 4} \cir{4, 4}
    \cir{0, 3} \fil{4, 3} \fil{0, 2} \cir{4, 2} \cir{0, 1} \fil{4, 1}
    \cir{0, 0} \cir{1, 0} \cir{2, 0} \cir{3, 0} \fil{4, 0}
    \fil{1.5, 2} \fil{2.5, 2}
    \rec{4}{0} \rec{4}{1} \rec{4}{2} \rec{4}{3} \rec{4}{4}
    \node at (2, -1.3) { \fns $9$ };
  \end{tikzpicture} \qquad
  \begin{tikzpicture}[scale=.45]
    \fil{0, 4} \fil{1, 4} \cir{2, 4} \fil{3, 4} \fil{4, 4}
    \cir{0, 3} \cir{4, 3} \fil{0, 2} \fil{4, 2} \cir{0, 1} \cir{4, 1}
    \cir{0, 0} \cir{1, 0} \cir{2, 0} \cir{3, 0} \cir{4, 0}
    \fil{1.5, 2} \fil{2.5, 2}
    \rec{0}{4} \rec{1}{4} \rec{2}{4} \rec{3}{4} \rec{4}{4}
    \node at (2, -1.3) { \fns $8$ };
  \end{tikzpicture} \qquad
  \begin{tikzpicture}[scale=.45]
    \cir{0, 4} \cir{1, 4} \fil{2, 4} \cir{3, 4} \cir{4, 4}
    \cir{0, 3} \cir{4, 3} \fil{0, 2} \fil{4, 2} \cir{0, 1} \cir{4, 1}
    \cir{0, 0} \cir{1, 0} \cir{2, 0} \cir{3, 0} \cir{4, 0}
    \fil{1.5, 2} \fil{2.5, 2}
    \rec{2}{4} \rec{0}{2} \rec{4}{2} \rec{1.5}{2} \rec{2.5}{2}
    \node at (2, -1.3) { \fns $5$ };
  \end{tikzpicture} \end{center}
\end{figure}
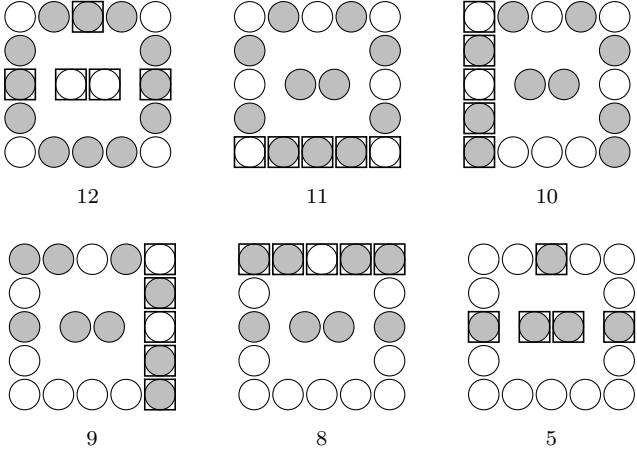

\section*{Conclusion}

We have initiated a study of the bit-flipping redundancy and developed
some first results, which may assist in the analysis of the decoding
failure rate for finite-length codes.  While the concept is similar to
the stopping set or the pseudocodeword redundancy, there seems to be
no obvious general connection, \newpage\noindent say in terms of upper
or lower bounds.  An open problem remains to discuss examples of codes
with infinite pseudoredundancy.  Also it would be interesting to
develop a redundancy for enhanced bit-flipping algorithms such as the
black-gray decoder, see e.g.\ \cite{DGK20, Vas21}.

\end{document}